\documentclass[11pt,letterpaper]{article}
\usepackage{times,mathptmx}
\DeclareMathAlphabet{\mathcal}{OMS}{cmsy}{m}{n}
\usepackage{fullpage}
\usepackage{amsmath,amsfonts,amssymb,amsthm,amscd}
\usepackage{hyperref,cite,url}
\hypersetup{pdfpagemode=UseNone,pdfstartview=}
\theoremstyle{plain} 
\newtheorem{theorem}{Theorem}[section]
\newtheorem{lemma}[theorem]{Lemma}

\theoremstyle{definition} 

\theoremstyle{remark} 

\newcommand{\vs}{\vspace{1.5mm}}
\newcommand{\G}{\mathbb{G}}
\newcommand{\Z}{\mathbb{Z}}
\newcommand{\bits}{\{0,1\}}
\newcommand{\mc}[1]{\mathcal{#1}}
\newcommand{\tb}[1]{\textbf{#1}}
\newcommand{\lb}{\linebreak[0]}

\title{Ciphertext Outdate Attacks on the Revocable Attribute-Based Encryption
Scheme with Time Encodings}

\author{
    Kwangsu Lee\footnote{Sejong University, Seoul, Korea.
    Email: \texttt{kwangsu@sejong.ac.kr}.}
}

\date{}

\begin{document}

\maketitle

\begin{abstract}
Cloud storage is a new computing paradigm that allows users to store their
data in the cloud and access them anytime anywhere through the Internet. To
address the various security issues that may arise in the cloud storage
accessed by a large number of users, cryptographic encryption should be
considered. Currently, researches on revocable attribute-based encryption
(RABE) systems, which provide user revocation function and ciphertext
update function by extending attribute-based encryption (ABE) systems that
provide access control to ciphertexts, are actively being studied.
Recently, Xu et al. proposed a new RABE scheme that combines ABE and
identity-based encryption (IBE) schemes to efficiently handle ciphertext
update and user revocation functionality. In this paper, we show that there
is a serious security problem in Xu et al.'s RABE scheme such that a cloud
server can obtain the plaintext information of stored ciphertexts by
gathering invalidated credentials of revoked users. Additionally, we also
show that the RABE scheme of Xu et al. can be secure in a weaker security
model where the cloud server cannot obtain any invalidated credentials of
revoked users.
\end{abstract}

\vs \noindent {\bf Keywords:} Cloud storage, Access control, Attribute-based
encryption, Revocation, Ciphertext update.

\section{Introduction}

Cloud storage is a computing paradigm that stores data in a centralized cloud
and allows users to access these data anytime anywhere on the Internet using
simple client devices. The main advantages of cloud storage include flexible
accessibility, ease management, and cost savings. Despite these advantages,
cloud storage is inevitably experiencing a variety of security issues because
it stores data in an external cloud storage that is outside of the control of
a data owner. The key reason that the cloud storage security differs from the
existing computer server security is that a cloud server that provides the
cloud storage service is not fully trusted so that the cloud server can
access the stored data and leak the sensitive information
\cite{KamaraL10,Ryan13}.

The easiest way to keep users data secure in cloud storage is to encrypt the
data and store it in the cloud. In this case, in order to share the encrypted
data with many users, it is needed to effectively control access to the
encrypted data according to the authority of the dynamically changing user.
That is, the cloud storage system needs to revoke some users whose
credentials are no longer valid so that revoked users cannot access data. In
addition, the cloud storage system should be able to prevent previously
revoked users to gain access to encrypted data that were created long ago by
using their old private keys after colluding with the cloud server.

To solve these problems in cloud storage, we can use attribute-based
encryption (ABE), which provides access control to ciphertexts. Boldyreva et
al. \cite{BoldyrevaGK08} proposed a revocable ABE (RABE) scheme that extends
the ABE scheme by providing the ability to revoke a user's private key. Sahai
et al. \cite{SahaiSW12} proposed a revocable-storage ABE (RS-ABE) scheme by
extending the concept of RABE that provides the ciphertext update
functionality to prevent previously revoked users from accessing previously
created ciphertexts in cloud storage. After that, Lee et al. \cite{LeeCLPY13}
proposed efficient RS-ABE schemes that can update ciphertexts more
efficiently by combining a self-updatable encryption (SUE) scheme and an ABE
scheme. Therefore, RS-ABE schemes, which provide user revocation and
ciphertext update, can be a solution to the problem of cloud storage
described above. Recently, Xu et al. \cite{XuYMD18} proposed an RABE scheme
that combines an ABE scheme with an IBE scheme by introducing new time
encoding functions to efficiently support ciphertext update than the existing
RS-ABE schemes. Compared with the most efficient RS-ABE scheme of Lee et al.,
the RABE scheme of Xu et al. is more efficient in terms of ciphertext size
and update key size.

In this paper, we show that it is possible to break the security of the RABE
scheme of Xu et al. \cite{XuYMD18}. A key feature of cloud storage is that a
cloud server is not fully trusted \cite{Ryan13}. In other words, the cloud
server faithfully performs the tasks requested by users, but is curious about
the information of the users' data. Thus the cloud server should also be
considered as an inside attacker. However, Xu et al. have overlooked that the
cloud server can be this type of attackers. Suppose that a cloud server is an
inside attacker in the RABE scheme of Xu et al. Then the cloud server first
gathers revoked credentials (revoked private keys) of users. Note that these
revoked credentials are usually no harm to the system since they are safely
disabled in publicly broadcasted key updates. Next, the cloud server derives
a new ciphertext associated with past time from stored original ciphertexts
in cloud storage without modifying the original ciphertexts. Then the cloud
server can sufficiently decrypt the derived ciphertext with past time by
using the revoked credentials and publicly available key updates. Thus the
RABE scheme of Xu et al. cannot be secure against this cloud server.

The organization of this paper is as follows: In Section 2, we first review
the RABE scheme of Xu et al. and their security model. Then, in Section 3, we
discuss our ciphertext outdate attack to the RABE scheme of Xu et al. by
exploiting the time encoding functions of Xu et al. Finally, we conclude in
Section 4.

\section{Revocable Attribute-Based Encryption}

In this section, we review the RABE scheme of Xu et al. \cite{XuYMD18} and
the security model of their RABE.

\subsection{Xu et al.'s Construction}

Before explaining the RABE scheme of Xu et al. \cite{XuYMD18}, we first
define the time encoding functions proposed by them. The \tb{TEncode}
function converts a time period $t$ to a bit string $bt$ of $\log_2 \mc{T}$
length by appending zero value to the prefix of the bit string. The
\tb{CTEncode} function converts a time period $t$ to an encoded bit string
$et$ by finding the first zero value and then converts all remaining values
to zero. The definitions of these two time encoding functions are described
follows:

\begin{description}
\item [\tb{TEncode}($t, \mc{T}$):] It takes a decimal number $t$. It
    encodes $t$ to a bit string $bt$. While $|bt| < \log_2 \mc{T}$, it
    performs $bt = 0 \| bt$. It returns the bit string $bt$.

\item [\tb{CTEncode}($t, \mc{T}$):] It takes a decimal number $t$. Let
    $[k]$ be the set $\{ 1, 2, \ldots, k\}$. It first sets an encoded
    string $et$ as empty one. It next obtains a bit string $bt$ by running
    \tb{TEncode}$(t, \mc{T})$ and sets $chk = false$. For each $j \in
    [\log_2 \mc{T}]$, it performs the following steps: if $bt[j] = 1$ and
    $chk = false$, then it sets $et[j] = 1$; otherwise it sets $chk = true$
    and $et[j] = 0$. It returns the encoded string $et$.
\end{description}

For example, we let the maximum time is $\mc{T} = 2^5$, and two time periods
are $t = 5$ and $t^* = 7$. In this case, the function \tb{TEncode}$(t = 5,
\mc{T})$ returns $bt = 00101$, the function \tb{TEncode}$(t^* = 7, \mc{T})$
returns $bt^* = 00111$, and the function \tb{CTEncode}$(t^* = 7, \mc{T})$
returns $et^* = 00000$.

The RABE scheme of Xu et al. follows the existing design methodology of
previous RABE schemes that combines an ABE scheme, a tree-based broadcast
scheme, and an IBE scheme in bilinear groups. In addition, Xu et al. have
changed the structure of ciphertext to provide ciphertext update
functionality by devising a new ciphertext encoding method. To use a
tree-based broadcast scheme, two additional functions \tb{Path} and
\tb{KUNodes} should be defined. The function \tb{Path} returns a set of nodes
in a binary tree that are in the path from the root node to the specified
leaf node, and the function \tb{KUNodes} returns a set of nodes that are root
nodes of sub-trees where the leaf nodes of all sub-trees can cover the set of
all non-revoked leaf nodes in the binary tree. It is required that if a leaf
node $v$ is not revoked then $\tb{Path}(v) \cap \tb{KUNode}(RL) \neq
\emptyset$ and if $v$ is revoked then $\tb{Path}(v) \cap \tb{KUNode}(RL) =
\emptyset$ where $RL$ is the set of revoked leaf nodes. We omit the detailed
descriptions of \tb{Path} and \tb{KUNodes} functions. For the more detailed
definition of these functions, see the work of Boldyreva et al.
\cite{BoldyrevaGK08}. The RABE scheme of Xu et al. is described as follows:

\begin{description}
\item [\tb{Setup}($1^\lambda, \mc{N}, \mc{T}, n$):] Let $\lambda$ be the
    security parameter, $\mc{N}$ be the maximum number of users, $\mc{T}$
    be the bounded system life time, and $n$ be the maximum number of
    attributes. It obtains a bilinear group $(p, \G, \G_T, e)$ by running
    $\mc{G}(1^\lambda)$ where $p$ is prime order of the groups. Let $g$ be
    a generator of $\G$.
    It selects a random exponent $\alpha$ and sets $g_1 = g^{\alpha}$. It
    also chooses random elements $g_2$, $\{ T_i \}_{i \in [n+1]}$, $U_0, \{
    U_j \}_{j \in [\log_2 \mathcal{T}]} \in \G$ and defines $T(x) =
    g_2^{x^n} \prod_{i=1}^{n+1} T_i^{\Delta_{i,[n+1]}(x)}$ where
    $\Delta_{i,J}(x) = \prod_{j \in J, j \neq i} \frac{x-j}{i-j}$.
    It sets a binary tree $BT$ with at least $\mc{N}$ number of leaves.
    Finally, it outputs a revocation list $RL = \emptyset$, a state $ST =
    BT$, a master key $MK = \alpha$, and public parameters %
    $PP = \big( (p, \G, \G_T, e), g, g_1, g_2, \{ T_i \}_{i \in [n+1]},
    U_0, \{ U_j \}_{j \in [\log_2 \mc{T}]} \big)$.

\item [\tb{GenKey}($id, \mathbb{A}, MK, ST, PP$):] Let $id$ be an identity
    and $\mathbb{A} = (M, \rho)$ be an access policy for attributes where
    $M$ is a $d \times \ell$ matrix. It assigns the user identity $id$ to a
    leaf node $\theta \in BT$.
    For each node $x \in \tb{Path}(\theta)$, it performs the following
    steps: 1) It fetches $\alpha_x$ from the node $x$. If $\alpha_x$ is not
    defined before, then it chooses a random $\alpha_x \in \Z_p$ and stores
    it in the node $x$. 2) Let $\vec{u}$ be a random $\ell$ dimensional
    vector over $\Z_p$ such that $1 \cdot \vec{u} = \alpha_x$. For each row
    $i$ in the matrix $M$, it chooses a random exponent $r_i$ and sets a
    partial private key $PSK_{id,x} = \big( \{ K_{i,0} = g_2^{M_i \cdot
    \vec{u}} T(i)^{r_i}, K_{i,1} = g^{r_i} \}_{i \in [d]} \big)$.
    Finally, it outputs a private key $SK_{id} = \big( \{ PSK_{id,x} \}_{x
    \in \tb{Path}(\theta)} \big)$ and an updated state $ST = BT$.

\item [\tb{UpdateKey}($t, RL, MK, ST, PP$):] Let $t$ be a revocation epoch
    and $RL$ be the revocation list. It obtains a bit string $bt$ by
    running \tb{TEncode}$(t, \mc{T})$. Let $\mc{V}_{bt}$ be the set of all
    $j$ for which $bt[j] = 0$.
    For each node $x \in \tb{KUNodes}(BT, RL, t)$, it performs the
    following steps: 1) It fetches $\alpha_x$ from the node $x$. If
    $\alpha_x$ is not defined before, then it chooses a random $\alpha_x
    \in \Z_p$ and stores it in the node $x$. 2) It chooses a random
    exponent $r$ and obtains a partial key update $PKU_{t,x} = \big( U_0 =
    g_2^{\alpha - \alpha_x} (U_0 \prod_{j \in \mc{V}_{bt}} U_j)^r, U_1 =
    g^r \big)$.
    Finally, it outputs a key update $KU_{t} = \big( \{ PUK_{t,x} \}_{x \in
    \tb{KUNodes}(BT,RL,t)} \big)$.

\item [\tb{DeriveDK}($SK_{id}, KU_{t}, PP$):] Let $SK_{id} = ( \{
    PSK_{id,x} \}_{x \in \tb{Path}(\theta)} )$ and $KU_t = ( \{ PKU_{t,x}
    \}_{x \in \tb{KUNodes}(BT,RL,t)} )$.
    If $\tb{Path}(\theta) \cap \tb{KUNodes}(BT,RL,t) = \emptyset$, then it
    outputs $\perp$. Otherwise, it finds a unique node $x \in
    \tb{Path}(\theta) \cap \tb{KUNodes}(BT,RL,t)$ and retrieves
    $PSK_{id,x}$ and $PKU_{t,x}$ for the node $x$ from $SK_{id}$ and
    $KU_{t}$ respectively.
    Finally it outputs a decryption key $DK_{id,t} = \big( \{ PSK_{id,x},
    PKU_{t,x} \} \big)$.

\item [\tb{Encrypt}($S, t, m, PP$):] Let $S$ be an attribute set, $t$ be
    time, and $m$ be a message. It obtains an encoded string $et$ by
    running \tb{CTEncode}$(t,\mc{T})$. Let $\mc{V}_{et}$ be the set of all
    $j$ for which $et[j] = 0$.
    It chooses a random exponent $s \in \Z_p$ and outputs an original
    ciphertext %
    $CT_{t} = \big( C = e(g_1, g_2)^s \cdot m, C_1 = g^s, \{ C_{2,i} =
    T(i)^s \}_{\rho(i) \in S}, E_1 = U_0^s, \{ E_{2,j} = U_j^s \}_{j \in
    \mc{V}_{et}} \big)$.

\item [\tb{UpdateCT}($CT_t, t', PP$):] Let $CT_t = (C, C_1, \{ C_{2,i} \},
    E_1, \{ E_{2,j} \}_{j \in \mc{V}_{et}} )$ be an original ciphertext for
    time $t$ and $t'$ be update time such that $t \leq t'$.
    If $t' < t$, then it returns $\perp$ to indicate that the time $t'$ is
    invalid. Otherwise, it obtains a bit string $bt$ by running
    \tb{TEncode}$(t, \mc{T})$. It chooses a random exponent $s' \in \Z_p$
    for randomization and outputs an updated ciphertext %
    $CT_{t'} = \big( C = C \cdot e(g_1, g_2)^{s'}, C_1 = C_1 \cdot g^{s'},
    \{ C_{2,i} = C_{2,i} \cdot T(i)^{s'} \}_{\rho(i) \in S}, E_{t'} = (C_1
    \prod_{j \in \mc{V}_{bt}} C_{2,j}) \cdot (U_0 \prod_{j \in \mc{V}_{bt}}
    U_{2,j})^{s'} \big)$.

\item [\tb{Decrypt}($CT_{t}, DK_{id,t}, PP$):] Let $CT_{t} = ( C, C_1, \{
    C_{2,i} \}, E_{t} )$ be an update ciphertext for time $t$ and
    $DK_{id,t} = ( PSK_{id,x}, PKU_{t,x} )$ be a decryption key where
    $PSK_{id,x} = ( \{ K_{i,0}, K_{i,1} \}_{i \in [d]} )$ and $PKU_{t,x} =
    ( U_0, U_1 )$.
    It computes a first component $A_1 = \prod_{\rho(i) \in S} (e(C_1,
    K_{i,0}) / e(C_{i,0}, K_{i,1}))^{w_i}$. Next, it computes a second
    component $A_2 = e(C_1, U_0) / e(E_t, U_1)$.
    It outputs a decrypted message $m$ by computing $C / (A_1 \cdot A_2)$.

\item [\tb{Revoke}($id, t, RL, ST$):] Let $id$ be an identity and $t$ be
    revocation time. It adds $(id, t)$ to $RL$ and returns the updated
    revocation list $RL$.
\end{description}

A cloud storage system consists of four entities: a trusted center, a cloud
server, a data owner, and a data user. The trusted center first runs
\tb{Setup} to obtain $MK$ and $PP$ and publishes $PP$. For each data user,
the trusted center runs \tb{GenKey} to generate each private key $SK_{id}$ of
each data user $id$. For each current time epoch $t$, the trusted center
periodically runs \tb{UpdateKey} to obtain a key update $KU_t$ for
non-revoked users and publishes $KU_t$. If a user with $id$ is revoked, then
the trusted center runs \tb{Revoke} to add this user to the revoked list. A
data owner who has a message $m$ can create an original ciphertext $CT_t$ at
time $t$ by running \tb{Encrypt} and then he securely sends $CT_t$ to the
cloud server for storing it in cloud storage. If a data user want to access
to the ciphertext in the cloud storage on time $t'$, then the cloud server
first computes an updated ciphertext $CT_{t'}$ by running \tb{UpdateCT} on
the original ciphertext and gives the updated ciphertext to the data user.
Next, the data user can decrypt the ciphertext $CT_{t'}$ by running
\tb{Decrypt} if he has a private key and his private key $SK_{id}$ is not yet
revoked in a key update $KU_{t'}$ on time $t'$.

\subsection{Security Model}

We describe the security model of the RABE scheme as defined by Xu et al.
\cite{XuYMD18}. The selective IND-RABE-CPA security is defined as the
following game between a challenger $\mc{C}$ and an adversary $\mc{A}$:

\vs\noindent \tb{Init}: $\mc{A}$ first submits a challenge attribute set
$S^*$.

\noindent \tb{Setup}: $\mc{C}$ generates an empty revocation list $rl$, a
state $ST$, a master key $MK$, and public parameters $PP$ by running the
setup algorithm \tb{Setup}$(\lambda, \mc{N}, \mc{T}, n)$, and then it gives
$PP$ to $\mc{A}$.

\noindent \tb{Phase 1}: $\mc{A}$ may adaptively request private key, key
update, and revocation queries to the following oracles.
    \begin{itemize}
    \item The private key generation oracle takes an identity $id$ and an
        access structure $\mathbb{A}$ as input, and returns a private key
        $SK_{id}$ by running \tb{GenKey}$(id, \mathbb{A}, MK, ST, PP)$.

    \item The key update oracle takes time $t$ as input, and returns a key
        update $KU_t$ by running \tb{UpdateKey}$(t, RL, \lb MK, ST, PP)$.

    \item The revocation oracle takes a revoked identity $id$ and time $t$
        as input, and updates the revocation list by running
        \tb{Revoke}$(id, t, RL, ST)$.
    \end{itemize}

\noindent \tb{Challenge}: $\mc{A}$ submits challenge time $t^* \in \mc{T}$
and two challenge messages $m_0^*, m_1^*$ of the same size with the following
constraints:
    \begin{itemize}
    \item If a private key for an identity $id$ and an access structure
        $\mathbb{A}$ such that $\mathbb{A}(S^*) = 1$ was queried to the
        private key generation oracle, then the revocation of the identity
        $id$ must be queried on time $t$ such that $t \leq t^*$ to the
        revocation oracle.

    \item If a non-revoked user with the identity $id$ whose access
        structure $\mathbb{A}$ satisfies the challenge attribute set $S^*$,
        then $id$ should not be previously queried to the private key
        generation oracle.
    \end{itemize}
$\mc{C}$ flips a random bit $b \in \bits$ and creates a challenge ciphertext
$CT^*$ by running \tb{Encrypt}$(S^*, t^*, m_b^*, PP)$.

\noindent \tb{Phase 2}: $\mc{A}$ continues to request private key, key
update, and revocation queries. $\mc{C}$ handles the queries as the same as
before and following the restrictions defined in the challenge phase.

\noindent \tb{Guess}: Finally $\mc{A}$ outputs a bit $b'$.

An RABE scheme is selectively IND-RABE-CPA secure if for any probabilistic
polynomial time adversary $\mc{A}$, the advantage of $\mc{A}$ in the above
RABE game defined as $\Pr[b = b'] - \frac{1}{2}$ is negligible in the
security parameter $\lambda$.

\section{Ciphertext Outdate Attack}

In this section, we show that there is an effective adversary against the
RABE scheme of Xu et al. \cite{XuYMD18}. To do this, we first analyze the
properties of two time encoding functions, \tb{TEncode} and \tb{CTEncode},
proposed by Xu et al. through the following two lemmas. The key to the
following two lemmas is that a challenge original ciphertext associated with
challenge time $t^*$ can be changed to a ciphertext element associated with
the past time $t$.

\begin{lemma} \label{lem:time-encode}
Let $st$ be a bit string in $\bits^{\log_2 \mc{T}}$ and $\mc{V}_{st}$ be the
set of all $j$ such that $st[j] = 0$. There exist time periods $t, t^* \in
\mc{T}$ such that $t < t^*$ and $\mc{V}_{bt} \subseteq \mc{V}_{et^*}$ where
$bt$ is obtained from $\tb{TEncode}(t, \mc{T})$ and $et^*$ is obtained from
$\tb{CTEncode}(t^*, \mc{T})$.
\end{lemma}

\begin{proof}
For the notational simplicity, we set $\mc{T} = 2^{\tau}$. To prove this
lemma, we first randomly choose time periods $t, t^*$ satisfying $0 < t < t^*
< 2^{\tau - 1}$. Then, we run \tb{TEncode}$(t, \mc{T})$ to get a bit string
$bt \in \bits^{\tau}$ and \tb{TEncode}$(t^*, \mc{T})$ to get another bit
string $bt^* \in \bits^{\tau}$. Since the time periods $t$ and $t^*$ are
smaller than $2^{\tau - 1}$, the first bit value $bt[1]$ and $bt^*[1]$ of the
two bit strings $bt$ and $bt^*$ have the same bit $0$. Now let 's analyze the
bit string $et^*$ obtained by running \tb{CTEncode}$(t^*, \mc{T})$. In the
\tb{CTEncode} algorithm, the algorithm finds the first position with a bit
value of 0 in the $bt^*$ bit string and then sets all subsequent bit values
to a value of zero. Thus, the resulting bit string $et^*$ becomes a bit
string with $0$ value in all positions since $bt^*[1] = 0$ is already fixed.
Therefore, the set $\mc{V}_{et^*}$ consists of $\{ 1, 2, \ldots, \tau \}$ and
the set of $\mc{V}_ {bt}$ should be a subset of $\{ 1, 2, \ldots, \tau \}$
since $0 < t$.
\end{proof}

As an example, let us look at the encoding results for two time periods $t =
5$ and $t^* = 7$ when the maximum time is $\mc{T} = 2^5$. Since $\mc{T} =
2^5$, the function \tb{TEncode}$(t = 5, \mc{T})$ returns a bit string $bt =
00101$, the function \tb{TEncode}$(t^* = 7, \mc{T})$ returns a bit string
$bt^* = 00111$, and the function \tb{CTEncode}$(t^* = 7, \mc{T})$ returns the
bit string $et^* = 00000$. Because $t = 5 < t^* = 7$ and $\mc{V}_{bt} = \{ 1,
2, 4 \} \subseteq \mc{V}_{et^*} = \{ 1, 2, 3, 4, 5 \}$, we can easily show
that two time periods $t = 5$ and $t^* = 7$ are one example of Lemma
\ref{lem:time-encode}.

\begin{lemma} \label{lem:cipher-outdate}
If there exist time periods $t, t^* \in \mc{T}$ such that $t < t^*$ and
$\mc{V}_{bt} \subseteq \mc{V}_{et^*}$ where $bt$ is obtained from
\tb{TEncode}$(t, \mc{T})$ and $et^*$ is obtained from \tb{CTEncode}$(t^*,
\mc{T})$, then a ciphertext element $E_t$ for time $t$ can be derived from an
original ciphertext $CT^*$ for time $t^*$.
\end{lemma}

\begin{proof}
As the same as in Lemma \ref{lem:time-encode}, we randomly choose time
periods $t$ and $t^*$ to satisfy $0 < t < t^* < 2^{\tau - 1}$. In the RABE
scheme of Xu et al., the original ciphertext $CT^*$ for the time $t^*$
includes ciphertext elements $E_1$ and $E_{2,j}$ for all $j \in
\mc{V}_{et^*}$. As shown in the previous Lemma \ref{lem:time-encode}, the set
$\mc{V}_{et^*}$ is defined as a set of all indices from $1$ to $\tau$ since
all bit values of $et^*$ are composed of $0$. In the description of the
\tb{UpdateCT} algorithm, the ciphertext element $E_t$ can be derived by
composing the elements $E_1$ and $E_{2,j}$ for all $j \in \mc{V}_{bt}$ of
$CT^*$. Therefore, it is possible to construct the element $E_t$ since
$\mc{V}_{bt} \subseteq \mc{V}_{et^*}$ is satisfied by Lemma
\ref{lem:time-encode}.
\end{proof}

By using the previous two lemmas, we show that a cloud server which stores
original ciphertexts generated by data owner can obtain sensitive information
of the original ciphertexts by gathering revoked credentials of revoked
users.

\begin{theorem}
There exists a probabilistic polynomial-time adversary that can break the
selective IND-RABE-CPA security of Xu et al.'s RABE scheme.
\end{theorem}

\begin{proof}
The main idea of our attack is for an inside adversary, which is a cloud
server, to derive an outdated ciphertext $CT_t$ of past time $t$ from the
original challenge ciphertext $CT^*$ of challenge time $t^*$ such that $t <
t^*$. If the adversary gathers revoked credentials of revoked users from the
Internet, then it can decrypt the outdated ciphertext by combining the
revoked credential with publicly available key updates.

A detailed adversary algorithm $\mc{A}$ breaking the RABE scheme of Xu et al.
is described as follows:
\begin{enumerate}
\item Initially $\mc{A}$ submits a challenge attribute set $S^*$ and
    receives public parameters $PP$.

\item After that, $\mc{A}$ gathers revoked credentials (private keys) of
    revoked users from the Internet. Note that revoked credentials of users
    can be available in Internet since they are safely revoked by the
    system although these credentials are intensionally revealed by an
    hacker or accidently revealed by a honest user.
    Let $SK_{id^*}$ be one of the obtained revoked private keys with an
    identity $id^*$ and an access policy $\mathbb{A}$ that satisfies
    $\mathbb{A}(S^*) = 1$. Because this private key $SK_{id^*}$ is revoked
    by a trusted center, we have that it is revoked on time $t^*$ in a key
    update $KU_{t^*}$, but not yet revoked in a key update $KU_{t}$ such
    that $t < t^*$. $\mc{A}$ sets past time $t$ and challenge time $t^*$
    that satisfies the condition described in Lemma \ref{lem:time-encode}.

\item Next $\mc{A}$ also gathers key updates from the Internet since each
    key update $KU_t$ is publicly broadcasted by the trusted center per
    each time period $t$. It now derive a decryption key $DK_{id^*,t}$ by
    combining the gathered private key $SK_{id^*}$ and the gathered key
    update $KU_{t}$ since $SK_{id^*}$ was not yet revoked on past time $t <
    t^*$. Note that it cannot derive $DK_{id^*,t^*}$ for the challenge time
    $t^*$ because $SK_{id^*}$ was already revoked on time $t^*$.

\item In the challenge step, $\mc{A}$ submits the challenge time $t^*$,
    randomly chosen two challenge messages $m_0^*, m_1^*$, and receives a
    challenge original ciphertext $CT^*$.
    Let $CT^* = (C, C_1, \{ C_{2,i} \}, E_1, \{ E_{2,j} \}_{j \in
    \mc{V}_{et^*}})$ be the original ciphertext on the time $t^*$ where
    $et^*$ is obtained from \tb{CTEncode}$(t^*, \mc{T})$. Consider the set
    $\mc{V}_{bt}$ for the fixed past time $t$. The ciphertext element $E_t$
    can be derived from the original ciphertext $CT^*$ by Lemma
    \ref{lem:cipher-outdate} because $\mc{V}_{bt} \subseteq \mc{V}_{et^*}$
    is satisfied. Thus $\mc{A}$ can easily derive an outdated ciphertext
    $CT_t = (C, C_1, {C_ {2, i}}, E_t)$ associated with the challenge
    attribute set $S^*$ and the past time $t$ by performing
    re-randomization.

\item Finally, $\mc{A}$ obtains the message $m^*$ by decrypting $CT_t$
    using $DK_{id^*,t}$ and outputs a bit $b'$ by comparing $m^*$ with the
    challenge messages.
\end{enumerate}

Now we analyze the success probability of the adversary $\mc{A}$ described
above. As shown above, the decryption succeeds because the behavior of
$\mc{A}$ satisfy the constraints of the security model and the outdated
ciphertext is also a valid ciphertext with the correct distribution.
Therefore, $\mc{A}$ wins the RABE security game since the advantage of
$\mc{A}$ is $1/2$.
\end{proof}

\subsection{Discussions}

In this section, we consider two possible defences against the ciphertext
outdate attacks.

\vs\noindent \tb{Weaker Security Model.} The reason for the above attack is
that a cloud server can decrypt a stored original ciphertext by deriving an
outdated ciphertext from the original ciphertext and using gathering revoked
credentials of revoked users. A naive approach to prevent this devastating
attack is to weaken the security model of RABE for cloud storage. In other
words, if an inside attacker such as a cloud server performs an attack, then
the inside attacker should be prohibited to gathering revoked credentials of
revoked users from the Internet. However, this weaker security model only
provides a limited security since it is very hard to forbid the inside
attacker to (passively) gathering some useful information from the Internet
\cite{PopaRZB11}.

\vs\noindent \tb{Self-Updatable Encryption.} Sahai et al. \cite{SahaiSW12}
first devised a new encryption scheme that supports ciphertext updates in
cloud storage by using the delegation feature of ABE. After that, Lee et al.
\cite{LeeCLPY13} proposed a self-updatable encryption scheme by combining a
binary tree with the key delegation feature of hierarchical identity-based
encryption. A self-updatable encryption scheme provides the most efficient
ciphertext update and it is proven to be secure against collusion attacks. In
order to enhance the security of Xu et al.'s RABE scheme, their RABE scheme
can be modified to use a self-updatable encryption scheme instead of using
time encoding functions. In addition, if a self-updatable encryption scheme
is used, a data owner does not need to use a secure channel when he sends a
ciphertext to cloud storage.

\section{Conclusion}

In this paper, we showed that a cloud server can perform the ciphertext
outdate attack to the RABE scheme of Xu et al. This attack was possible
because the cloud server is not a fully trusted entity and it can derive
another ciphertext associated with past time from the original ciphertext
stored by a user. Although Xu et al. showed that their RABE scheme is secure
against outside attackers, they didn't consider the cloud server can be an
inside attacker in the security proof. One naive defence against this attack
is to consider a weaker security model where an insider attacker who has
access to the original ciphertext cannot obtain revoked credentials of users
by preventing the inside attacker from accessing to the Internet. A better
defence for this attack is to use the self-updatable encryption scheme of Lee
et al. \cite{LeeCLPY13} for efficient ciphertext updating since an RABE
scheme that uses a self-updatable encryption scheme can be secure against
collusion attacks.

\bibliographystyle{plain}
\bibliography{cipher-outdate-attack}

\end{document}